\documentclass[12pt,a4paper]{article}
\usepackage[utf8]{inputenc}
\usepackage{color}  
\definecolor{coolblack}{rgb}{0.0, 0.18, 0.39}
\definecolor{midnightblue}{rgb}{0.1, 0.1, 0.44}
\definecolor{prussianblue}{rgb}{0.0, 0.19, 0.33}
\definecolor{oxfordblue}{rgb}{0.0, 0.13, 0.28}
\usepackage{hyperref,float}
\hypersetup{
    colorlinks=true,
    linkcolor=blue,
    filecolor=magenta,      
    urlcolor=blue,
    citecolor =blue,
}

\usepackage{amsmath,amsthm}
\usepackage{amsfonts}
\usepackage{amssymb}
\usepackage{graphicx}
\usepackage{lmodern}
\usepackage[round]{natbib}
\usepackage{setspace}
\usepackage{booktabs}
\usepackage{wrapfig}
\usepackage{textcomp}
\usepackage{pdflscape}
\usepackage[small,bf]{caption}
\usepackage{breqn}
\usepackage{lscape}
\usepackage{pgfplots}
\usepackage{grffile}
\usepackage{adjustbox}

\usepackage{tikz}
\usetikzlibrary{snakes}

\newcommand{\bX}{\ensuremath{\bm{X}}}
\newcommand{\IZ}{\ensuremath{\mathbb{Z}}}

\newcommand{\IN}{\ensuremath{\mathbb{N}}}

\newcommand{\bPhi}{\ensuremath{\boldsymbol\Phi}}
\newcommand{\bPsi}{\ensuremath{\boldsymbol\Psi}}
\newcommand{\btheta}{\ensuremath{\boldsymbol\theta}}
\newcommand{\bepsilon}{\ensuremath{\boldsymbol\epsilon}}
\newcommand{\bbeta}{\ensuremath{\boldsymbol\eta}}
\newcommand{\bSigma}{\ensuremath{\boldsymbol\Sigma}}

\newcommand{\mC}{\ensuremath{\mathcal{C}}}

\newtheorem{prop}{Proposition}
\newenvironment{prop*}
  {\ex}
  {\endex}

\newtheorem{remark}{Remark}
\newenvironment{remark*}
  {\ex}
  {\endex}

\newenvironment{definition*}
  {\ex}
  {\endex}
  
\usepackage{cleveref}
\usepackage{multirow}
\usepackage{subcaption}
\usepackage[stable]{footmisc}
\usepackage{bm}
\usepackage{lineno}

\usepackage[left=1in,right=1in,top=1.5in,bottom=1.5in]{geometry}

\author{Jozef \textsc{Barun\'{i}k}$^{\ddag}$ and Michael \textsc{Ellington}$^{\dag}$}
\title{Persistence in Financial Connectedness and Systemic Risk\thanks{\scriptsize{We thank Oliver Linton, Wolfgang H\"{a}rdle, Melanie Schienle, Catherine Forbes, Lubos Hanus, Luk\'{a}\v{s} V\'{a}cha, Ryland Thomas, Chris Florackis, Costas Milas, Alex Kostakis and Charlie Cai for invaluable discussions and comments. We are grateful to Lubo\v{s} Hanus for help in furnishing and converting estimation codes. We acknowledge insightful comments from numerous seminar presentations, such as: the 2019 and 2020 Society for Economic Measurement Conferences; the Danish National Bank; the 2019 STAT of ML conference; the 13${\text{th}}$ International Conference on Computational and Financial Econometrics; and many more. Jozef Barun\'{i}k gratefully acknowledges support from the Czech Science Foundation under the 19-28231X (EXPRO) project. For estimation of dynamic horizon specific networks, we provide packages \texttt{DynamicNets.jl} in \textsf{JULIA} and \texttt{DynamicNets} in \textsf{MATLAB}. The packages are available at \url{https://github.com/barunik/DynamicNets.jl} and \url{https://github.com/ellington/DynamicNets}.
\noindent \textbf{Disclosure Statement:} Jozef Barun\'{i}k and Michael Ellington have nothing to disclose.}}}
\date{\today}

\begin{document}

\begin{titlepage}
\maketitle 

\begin{abstract}
\noindent This paper characterises dynamic linkages arising from shocks with heterogeneous degrees of persistence. Using frequency domain techniques, we introduce measures that identify smoothly varying links of a transitory and persistent nature. Our approach allows us to test for statistical differences in such dynamic links. We document substantial differences in transitory and persistent linkages among US financial industry volatilities, argue that they track heterogeneously persistent sources of systemic risk, and thus may serve as a useful tool for market participants.
\end{abstract}

\noindent JEL Classifications: C10, C40, C55, C58, G00 \\
\noindent Keywords: Finance, Network connections, Variance Decompositions, Persistence, Spectral Domain.

\begin{small}
\begin{singlespace}
\noindent\rule{10cm}{0.4pt}\\
\ddag \textit{Institute of Economic Studies, Charles University, Opletalova 26, 110 00, and The Czech Academy of Sciences, IITA, Pod Vodárenskou Věží 4, 182 08, Prague, Czech Republic}.\\
\href{mailto: barunik@fsv.cuni.cz}{barunik@fsv.cuni.cz}\\
\dag \textit{University of Liverpool Management School, Chatham Building, Liverpool, L69 7ZH, UK}.
\href{mailto: m.ellington@liverpool.ac.uk}{m.ellington@liverpool.ac.uk}
\end{singlespace}
\end{small}

\end{titlepage}

\newpage

\section{Introduction}

Firms and economic units create connections through a variety of channels \citep[see e.g.][]{richmond2019trade,garvey2015analytical,herskovic2020firm}.\footnote{\cite{richmond2019trade} measures connections through consumption growth. \cite{garvey2015analytical} tracks connections that describe the supply chain, and \cite{herskovic2020firm} studies network connections between firm volatilities.} These connections are dynamic and vary over time with changing states of an economy, as both stable and uncertain periods are associated with different shocks. At the same time, an increasing number of authors argue that economic variables are driven by shocks that influence their future value with heterogeneous levels of persistence \citep{bandi2019spectral,dew2016asset}. Connectedness and the subsequent network structures that emerge from these relationships are central to risk measurement and management, as well as to understanding macroeconomic risks that emerge over business cycles. 

The main objective of this paper is to introduce measures capable of identifying the smoothly varying persistence structure of such linkages. As an interesting and important example, we provide an analysis of US sectors, with a focus on the financial sector, and identify a heterogeneous persistence structure of systemic risk.

There are a variety of ways to measure connectedness.\footnote{Some propose measures derived from correlations or coefficient estimates \citep[see e.g.][]{engle2012dynamic,geraci2018measuring,calabrese2019new}, while others track links between individual companies and broader market/economic movements \citep[see e.g.][]{acharya2012capital,adrian2016covar,acharya2017measuring}
.} In particular, \cite{diebold2014} provides a unifying framework for measuring connectedness and other network properties using variance decompositions from an approximate model. Variance decompositions track how shocks affect the future variation of variables within a system, and are therefore a natural choice for inferring network connectedness from data.\footnote{The subsequent literature has widely adopted the \cite{diebold2014} approach to address a variety of issues in finance and economics \citep[e.g.][]{yang2017quantitative,barunik2020asymmetric}.} However, the time evolution of such measures typically involves the estimation of static models that require covariance stationarity and roll through a time series of data. More importantly, such measures aggregate shocks and mask the persistence structures of a network.

In this paper, we provide the tools to identify dynamic connectedness and other key network measures from variance decompositions. We argue that different levels of connectedness can form around transitory and persistent components of shocks within a system of financial data. To capture the time-varying nature of such transitory and persistent components of connectedness, we further consider time-varying variance decomposition matrices from vector autoregressions (VARs) as dynamic adjacency matrices. To identify the persistence structure of the network, we propose to use localised spectral decompositions\footnote{Note that frequency domain techniques are useful tools for denoising \citep{sun2012new,haven2012noising} and forecasting \citep{sevi2014forecasting,barunik2016modeling} financial time series.} of variance error forecasts.

The main contribution of our paper is to provide a novel framework for measuring dynamic relationships that relate to different horizons of interest in multivariate time series models. We use a locally stationary Bayesian time-varying parameter VAR model, which is readily available in high-dimensional settings. We also develop a test for differences in connectedness over different horizons and show how to infer differences over time. We provide Monte Carlo evidence that our measures are able to reliably track connections from different data generating processes (DGPs), including those that are non-Gaussian. Finally, we make computationally efficient packages 
\texttt{DynamicNets.jl} in \textsf{JULIA} and \texttt{DynamicNets} in \textsf{MATLAB} 
that allows one to obtain our measures on data the researcher desires.\footnote{The packages are available at \url{https://github.com/barunik/DynamicNets.jl} and \url{https://github.com/ellington/DynamicNets}}

Our approach provides a solution to the problems of using rolling windows \citep[see e.g.][]{demirer2018estimating} that does not suffer from dimensionality issues or inference problems. The Bayesian nature of our framework incorporates prior shrinkage and provides information about estimation uncertainty from the posterior distribution of the connectedness measures. This is in sharp contrast to conventional studies that only provide point estimates and rely on bootstrapping for confidence intervals. Our measures are also readily available for applications with large data systems. This extends the pairwise approach in \cite{geraci2018measuring} to study linkages between firms. 

The linkages that form over different horizons with heterogeneous persistence are important for a number of reasons. First, economic theory suggests that the marginal utility of agents' preferences depends on cyclical components of consumption \citep[see e.g.][]{giglio2015very,bandi2017business} and also on investment horizons in their risk attitudes \citep{dew2016asset}. Such behaviour can be observed, for example, under myopic loss aversion, where an agent's decision depends on the valuation horizon. 

Second, unanticipated shocks or news have the capacity to alter these preferences and can therefore generate transitory and persistent linkages of different strengths. For example, a shock that has an impact at longer horizons may reflect permanent changes in expectations of future price movements. Such a shock may lead to a permanent change in a firm's future dividend payments \citep{balke2002low}. Conversely, a shock that affects shorter horizons may suggest temporary changes in future price movements. For example, suppose that the shock relates only to a change in an upcoming dividend payment. This would likely result in a very short term change, reflecting the transitory nature of the news. 

Third, firms have different short-run and long-run objectives, and investors view short-run and long-run risks differently \citep[see e.g.][]{drechsler2011s,gerrard2022optimal}. This behaviour motivates the long-run risk asset pricing literature pioneered by \cite{bansal2004risks} and \cite{bansal2010long}. The implication here is that investment horizons may be a source of systematic risk that investors demand compensation for \citep[contributions on this topic include e.g.][]{brennan2018capital,chaudhuri2019dynamic}.  

Identifying network structures that form due to idiosyncratic shocks is also relevant because they can determine aggregate fluctuations \citep{acemoglu2012network}. These links between individual or firm-level entities create systemic risks for sectors and the economy as a whole \citep[][]{billio2012econometric,acemoglu2017microeconomic}. Such risks drive changes in uncertainty, which can be key factors in business cycles and financial distress \citep{bloom2018really}. \cite{gabaix2011granular} shows that sectoral co-movements are salient features of business cycles. Meanwhile, \cite{atalay2017important} finds that the lion's share of variation in output growth is due to idiosyncratic industry-level shocks.

However, an understanding of the potential longevity of risks arising from these linkages and their importance in driving financial turmoil or business cycles is incomplete. We show how our approach provides measures of transitory and persistent linkages using the daily realised firm-level volatilities of S\&P500 constituents. We classify the constituents into their eleven main sectors according to the Global Industry Classification Standard (GICS) and measure network connections from the transitory and persistent components of volatility shocks. 

Our empirical results document substantial heterogeneities in transient and persistent measures of connectedness that reveal the nature of systemic risks arising from networks. Specifically, we document: i) spikes in persistent network connectedness when long-lasting financial and economic events occur; and ii) statistically significant differences between transitory and persistent network connectedness across sectors. 
Our network measures can serve as an online monitoring tool for sectoral uncertainty in markets of interest to macroprudential supervisors and investors alike.


The rest of the paper proceeds as follows. Section \ref{empirical_measures} derives our measures from locally stationary processes, discusses estimation, and proposes a test procedure for statistical differences in transitory and persistent connectedness. Section \ref{Simulations} provides Monte Carlo evidence that our measures are able to reliably track linkages and correctly identify statistical differences. In Section \ref{sec:empirical}, we examine the links between firm-level volatilities of S\&P500 sector constituents and assess the information content of sector connectedness measures beyond leading measures of uncertainty. Finally, section \ref{conclusion} concludes.

\section{Measuring Transitory and Persistent Connections}\label{empirical_measures}

Here we show how one can measure connectedness using the time-varying spectral decompositions. Our measures of connectedness are based on locally stationary processes. This assumes that the process is approximately stationary over a short time interval, which allows us to incorporate time variation into our analysis. This in turn allows us to construct our measures of frequency-dependent, time-varying network connectedness. 

Formally, consider a doubly indexed $N$-variate time series $(\bX_{t,T})_{1\le t \le T,T \in \IN}$ with components $\bX_{t,T}=(\bX_{t,T}^1,\ldots,\bX_{t,T}^N)^{\top}$ that describe all variables in an economy. Here $t$ refers to a discrete time index and $T$ is an additional index indicating the sharpness of the local approximation of the time series $(\bX_{t,T})_{1\le t \le T,T \in \IN}$ by a stationary one. Coarsely speaking, we can consider $(\bX_{t,T})_{1\le t \le T,T \in \IN}$ to be a weakly locally stationary process if, for a large $T$, given a set $S_T$ of sample indices such that $t/T\approx u$ over $t\in S_T$, the sample $(\bX_{t,T})_{t \in S_T}$ approximates the sample of a weakly stationary time series depending on the rescaled location $u$. Note that $u$ is a continuous time parameter referred to as the rescaled time index, and $T$ is interpreted as the number of available observation, hence $1\le t \le T$ and $u\in[0,1]$, see \cite{dahlhaus1996kullback} for further details.

We assume that the economy follows a locally stationary TVP-VAR of lag order $p$ as
\begin{equation}\label{eq:VAR1}
\bX_{t,T}=\bPhi_{1}(t/T)\bX_{t-1,T}+\ldots+\bPhi_{p}(t/T)\bX_{t-p,T} + \bepsilon_{t,T},
\end{equation}
where $\bepsilon_{t,T}=\bSigma^{-1/2}(t/T)\bbeta_{t,T}$ with $\bbeta_{t,T}\sim NID(0,\boldsymbol{I}_M)$ and $\bPhi(t/T)=(\bPhi_{1}(t/T),\ldots,\bPhi_{p}(t/T))^{\top}$ are the time varying autoregressive coefficients. In a neighborhood of a fixed time point $u_0=t_0/T$, we approximate the process $\bX_{t,T}$ by a stationary process $\widetilde{\bX}_t(u_0)$ as
\begin{equation}
\widetilde{\bX}_t(u_0)=\bPhi_1(u_0)\widetilde{\bX}_{t-1}(u_0)+\ldots+\bPhi_p(u_0)\widetilde{\bX}_{t-p}(u_0) + \bepsilon_t,
\end{equation}
with $t\in \IZ$ and under suitable regularity conditions $|\bX_{t,T} - \widetilde{\bX}_t(u_0)| = O_p\big( |t/T-u_0|+1/T\big)$ which justifies the notation ``locally stationary process.'' Crucially, the process has time varying VMA($\infty$) representation \citep{dahlhaus2009empirical,roueff2016prediction}
\begin{equation}
\bX_{t,T} = \sum_{h=-\infty}^{\infty} \bPsi_{t,T}(h)\bepsilon_{t-h}
\end{equation}
where $\bPsi_{t,T}(h) \approx\bPsi(t/T,h)$ is a stochastic process satisfying $\sup_{\ell} ||\bPsi_t-\bPsi_{\ell}||^2 = O_p(h/t)$ for $1\le h \le t$ as $t\rightarrow \infty$. Specifically,   $\bPsi_{t,T}(h)=\left[\bPhi_{t,T}(h)\right]^{-1}$, which is key to understanding dynamics. Since $\bPsi_{t,T}(h)$ contains an infinite number of lags, we approximate the moving average coefficients at $h=1,\ldots,H$ horizons (see the detailed discussion below). The network characteristics rely on variance decompositions, which are transformations of the impulse response functions, $\bPsi_{t,T}(h)$, and permit the measurement of the contribution of shocks to the system.

Since a shock to a variable in the model does not necessarily appear alone, i.e. orthogonally to shocks to other variables, an identification scheme is crucial in calculating variance decompositions. We adapt the generalized identification scheme in \cite{pesaran1998generalized} to locally stationary processes. A natural way to disentangle connections that form over transitory and persistent components of shocks is to consider a spectral representation of the approximating model.\footnote{\cite{barunik2018measuring} disentangle long-run and short-run unconditional network connections using standard VAR models.} Hence instead of impulse responses, we propose to use the (local) frequency response of a shock. The building block of our measures consider a time-varying frequency response function $\bPsi_{t,T}(e^{-i\omega}) = \sum_h e^{-i\omega h} \bPsi_{t,T}(h)$ which we obtain from a Fourier transform of the coefficients with $i=\sqrt{-1}$.

Before introducing our network measures, we define the time varying spectral density of $\bX_{t,T}$ at frequency $\omega$ which is locally the same as the spectral density of $\widetilde{\bX}_t(u)$ at $u=t/T$ as a Fourier transform of VMA($\infty$) filtered series as
\begin{equation}
\boldsymbol S_{\bX}(u,\omega) = \sum_{h=-\infty}^{\infty}\mathbb{E}\Big[ \widetilde{\bX}_{t+h}(u)\widetilde{\bX}_{t}^{\top}(u)  \Big]e^{-i\omega h}=\Big\{\bPsi(u,e^{-i\omega})\Big\}\bSigma(u)\Big\{\bPsi(u,e^{+i\omega})\Big\}^{\top}.
\end{equation}

The time-varying spectral density is a key quantity for understanding frequency dynamics. It describes the distribution of the time varying covariance of $\bX_{t,T}$ over frequency components $\omega$. The local spectral density captures the influence of the time-varying parameters through the impulse transfer functions $\bPsi(u,e^{-i\omega}),\: \bPsi(u,e^{+i\omega}) $ above. Using the spectral representation for the local covariance that is associated with the local spectral density,
\begin{equation}
\mathbb{E}\left[ \widetilde{\bX}_{t+h}(u)\widetilde{\bX}_{t}^{\top}(u)\right] = \int_{-\pi}^{\pi} \boldsymbol S_{\bX}(u,\omega) e^{i\omega h} d \omega
\end{equation}
we can naturally introduce time-varying frequency domain counterparts of variance decompositions. This is important since, as \cite{diebold2014} note, we can view the variance decomposition matrix as an adjacency matrix forming asymmetric connections among a system of variables. In our case, this allows us to define dynamic adjacency matrices with different degrees of persistence. 

\subsection{A Route Towards Transitory and Persistent Connectedness}

Variance decompositions are transformations of impulse responses $\bPsi_{t,T}(h)$ that allow us to measure the contribution of shocks to the system and thus to characterise the networks that form in response to shocks. While it may seem natural to choose a forecast horizon $h=1,\ldots,H$ of interest, this choice is costly in terms of information aggregation and hence loss. In contrast to the cumulative information with increasing $h$, the spectral representation of the impulse responses $\bPsi_{t,T}(e^{-i\omega})$ contains much richer and more precise information. Switching to the frequency domain allows one to trace network connections arising from transitory and persistent components of shocks.

To illustrate, consider a simple bivariate system in which the links between two variables are of interest. Specifically, we are interested in how variable $b$ responds to two different shocks to variable $a$.\footnote{In this example, we assume that the impact of own shocks on variable $b$ increases by 1 unit and is persistent. This means that the impact of a 1 unit shock to variable $b$ affects the value of $b$ up to 20 horizons after we observe the shock. Note that our analysis in Figure \ref{fig:FEVD} remains the same if the own shocks are transitory.} The first is a \textit{transitory shock} that causes the variable $b$ to rise by one unit in period 1 and fall by one unit in period 2, before returning to zero from period onwards. Here we expect the proportion of error variation to be large at short horizons and small at long horizons due to the purely transitory effect. 

The second is a \textit{persistent shock} that results in a unit increase in the variable $b$ and a gradual decrease to zero over the impulse horizon. In this case, the fraction of error variation is expected to be large at longer horizons and smaller at shorter horizons. How these quantities in the time and frequency domain capture the responses to these two shocks is the main motivation for our measures.

\begin{figure}[!ht]
\centering
\scalebox{0.9}{\includegraphics{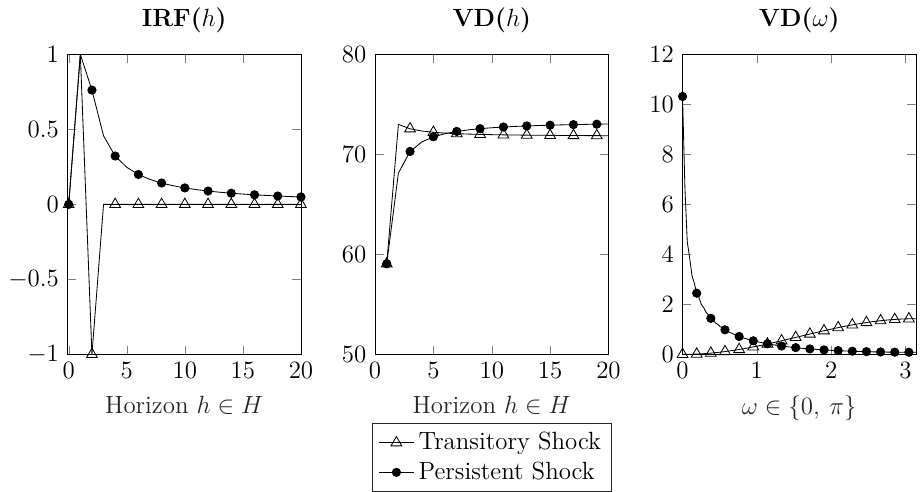}}
    \caption{\textbf{Impulse Response Functions and Variance Decompositions in time and frequency domain} \\ \small{Notes: This figure reports impulse response functions (left), variance decompositions in time domain (middle) and the corresponding spectral representation of the forecast error variance decompositions (right), of a variable $b$ within a bivariate system with respect to a shock in the variable $a$. We consider two types of shocks, a transitory shock (triangles) and a persistent shock (dots) impacting horizon $h=1,\ldots,H$ and frequency $\omega \in \{ 0,\ldots, \omega \}$.}}
      \label{fig:FEVD}
\end{figure}

Figure \ref{fig:FEVD} shows the impulse response functions of variable $b$ to shocks to variable $a$ in the left panel, the corresponding fraction of forecast error variances in the middle panel, and their corresponding spectral decompositions of forecast error variance fractions in the right panel. Note that one can hardly identify persistence of the shock from the almost indistinguishable forecast error variance shares (variance decompositions) in the time domain depicted by the middle panel of figure \ref{fig:FEVD}. A persistent shock results in a slightly larger value of forecast error variation relative to the variation due to a transitory shock. At the same time, if we were to estimate these quantities from the data and take into account the uncertainty of the estimates (the figure plots the theoretical values of a simple example), they would become statistically indistinguishable.

In contrast, the spectral representation of the forecast error variance shares in the right-hand panel of Figure \ref{fig:FEVD} accurately captures the heterogeneous impact of the two shocks across frequencies. The transitory shock in the variable $a$ has a negligible impact at low frequencies (i.e. close to $\omega=0$), indicating that this shock has no importance for the long-run variation of the variable $b$, and larger weights at higher frequencies, showing the transitory nature of the link established by this shock. Conversely, the persistent shock affects low frequencies and correctly identifies a persistent link between the variables.

The main implication is that we can construct network measures that take into account the nature of the shocks that form such links. Thus, using spectral decompositions, we are able to identify transitory and persistent network links that are not apparent in the time domain.

\subsection{Measuring Connectedness}

The following proposition establishes the time-varying spectral representation of the variance decomposition of shocks from variable $j$ to variable $k$. This is central to the existence of network measures in the time-frequency domain.\footnote{Note to notation: $[\boldsymbol A]_{j,k}$ denotes the $j$th row and $k$th column of matrix $\boldsymbol A$ denoted in bold. $[\boldsymbol A]_{j,\cdot}$ denotes the full $j$th row; $[\boldsymbol A]_{\cdot,j}$ denotes the full $j$th column. A $\sum A$, where $A$ is a matrix that denotes the sum of all elements of the matrix $A$.} 

\begin{prop}[Dynamic Adjacency Matrix]
	\label{prop:3}
	Suppose $\mathbf{\bX_{t,T}}$ is a weakly locally stationary process with $\sigma_{kk}^{-1} \displaystyle\sum_{h = 0}^{\infty}\left|\Big[\bPsi(u,h) \bSigma(u) \Big]_{j,k}\right| < +\infty, \forall j,k.$ Then the \textbf{time-frequency variance decompositions} of the $j$th variable at a rescaled time $u=t_0/T$ due to shocks in the $k$th variable on the frequency band $d = (a,b): a,b \in (-\pi, \pi), a < b$ form a \textbf{dynamic adjacency matrix} defined as
	\begin{equation}
	\Big[ \btheta(u,d) \Big]_{j,k} = \frac{\sigma_{kk}^{-1} \displaystyle \int_{a}^{b} \left| \bigg[ \bPsi(u,e^{-i \omega}) \bSigma(u) \bigg]_{j,k} \right|^2 d \omega}{ \displaystyle \int_{-\pi}^{\pi} \Bigg[ \Big\{\bPsi(u,e^{-i \omega}) \Big\}\bSigma(u) \Big\{ \bPsi(u,e^{+i \omega }) \Big\}^{\top}  \Bigg]_{j,j} d \omega}\end{equation}
	where $\bPsi(u,e^{-i\omega}) = \sum_h e^{-i\omega h} \bPsi(u,h)$ is local impulse transfer function or frequency response function computed as the Fourier transform of the local impulse response $\bPsi(u,h)$
\end{prop}
\begin{proof}
	See Appendix A.
\end{proof}

It is important to note that $\Big[ \btheta(u,d) \Big]_{j,k}$ is a natural disaggregation of traditional variance decompositions to time-varying frequency bands. This is because a portion of the local error variance of the $j$th variable at a given frequency band due to shocks in the $k$th variable is scaled by the variance of the $j$th variable. Note that while the Fourier transform of the impulse response generally takes on complex values, the quantity in proposition (\ref{prop:3}) is the squared modulus of weighted complex numbers, thus producing a real quantity.

This relationship is an identity which means the integral is a linear operator, summing over disjoint intervals covering the entire range $(-\pi, \pi)$ recovers the time domain counterpart of the local variance decomposition for $h\rightarrow \infty$. The following remark formalizes this fact.

\begin{remark}[Aggregation of Adjacency Matrix]
	\label{rem:recomposition}
	Denote by $d_s$ an interval on the real line from the set of intervals $D$ that form a partition of the interval $(-\pi, \pi)$, such that $\cap_{d_s \in D} d_s = \emptyset, $ and $\cup_{d_s \in D} d_s = (-\pi, \pi)$. Due to the linearity of integral and the construction of $d_s$, we have 
	$$
	\Big[ \btheta(u) \Big]_{j,k} = \sum_{d_s \in D} \Big[ \btheta(u,d_s) \Big]_{j,k}.$$
\end{remark}

Remark (\ref{rem:recomposition}) is important as it establishes the aggregation of network connectedness measures across different frequency bands to its time domain, \textit{total} counterpart. Hence one can easily obtain time varying network measures across any horizon of interest using frequency bands that will always sum up to an aggregate time domain counterpart. 

As the rows of the time-frequency network connectedness do not necessarily sum to one, we normalize the element in each by the corresponding row sum

\begin{equation}
\Big[ \widetilde \btheta(u,d) \Big]_{j,k} = \Big[ \btheta(u,d) \Big]_{j,k}\Bigg/ \sum_{k=1}^N\Big[  \btheta(u) \Big]_{j,k}
\end{equation}

Our notion that we can approximate well the process $\bX_{t,T}$, by a stationary process $\widetilde{\bX}_t(u)$ in a neighborhood of a fixed time point $u=t/T$, means that all associated local quantities approximate well their time varying counterparts. Following the arguments in \cite{dahlhaus1996kullback}, and using mild assumptions, one can easily see that local variance decompositions at a frequency band $\widetilde \btheta(u,d)$ approximate well the time-varying variance decompositions of the process $\bX_{t,T}$. 

Note that the local generalized variance decompositions form a dynamic adjacency matrix that defines a time-varying network at a given frequency band. Thus, we can directly use our measures as time-varying network characteristics that contain richer information in comparison to typical network analysis. In our notion, variance decompositions can be viewed as weighted links showing the strengths of connections. In addition, the links are directional, meaning that the $j$ to $k$ link is not necessarily the same as the $k$ to $j$ link, and hence the adjacency matrix is asymmetric. Even more important, the adjacency matrix is time-varying and frequency specific that allows the study of time-varying network characteristics at various frequency bands of the user's choice. The simplest is to measure transitory network connections over the short-run and persistent ones over the long-run.

Now we can define network connectedness measures that characterize a time-varying and frequency specific network. We define local network connectedness measures at a given frequency band as the ratio of the off-diagonal elements to the sum of the entire matrix

\begin{equation}\label{network}
\mC(u,d) = 100\times\displaystyle \sum_{\substack{j,k=1\\ j\ne k}}^N \Big[\widetilde \btheta(u,d)\Big]_{j,k}\Bigg/\displaystyle \sum_{j,k=1}^N \Big[\widetilde \btheta(u)\Big]_{j,k}
\end{equation}
This measures the contribution of forecast error variance attributable to all shocks in the system, minus the contribution of own shocks over frequency band $d$ and infers system-wide connectedness over such frequency band. We can also define measures that reveal when an individual variable in the economy is a transmitter or receiver of shocks. Local directional connectedness measures how much of each variables's $j$ variance is due to shocks in other variables $k\ne j$ in the economy over frequency band $d$ is given by
\begin{equation}\label{from}
\mC_{j\leftarrow\bullet}(u,d) = 100\times\displaystyle \sum_{\substack{k=1\\ k\ne j}}^N \Big[\widetilde \btheta(u,d)\Big]_{j,k}\Bigg/\displaystyle \sum_{j,k=1}^N \Big[\widetilde \btheta(u)\Big]_{j,k},
\end{equation}
defining the so-called \textsc{from} connectedness. One can precisely interpret this quantity as from-degrees (often called out-degrees in the network literature) associated with the nodes of the weighted directed network represented by the variance decompositions matrix generalized to a time-varying frequency specific quantity. Likewise, the contribution of variable $j$ to variances in other variables is computed as 
\begin{equation}\label{to}
\mC_{j\rightarrow \bullet}(u,d) = 100\times\displaystyle \sum_{\substack{k=1\\ k\ne j}}^N \Big[\widetilde \btheta(u,d)\Big]_{k,j}\Bigg/\displaystyle \sum_{k,j=1}^N \Big[\widetilde \btheta(u)\Big]_{k,j}
\end{equation}

and is the so-called \textsc{to} connectedness. Again, one can precisely interpret this as to-degrees (often called in-degrees in the network literature) associated with the nodes of the weighted directed network represented by the variance decompositions matrix. These two measures show how other variables contribute to the variation of variable $j$, and how variable $j$ contributes to the variation of others, respectively, in a time varying fashion at a chosen frequency band. We note here that taking the difference between \textsc{to} connectedness and \textsc{from} connectedness summarizes information regarding directional connections in net-terms. Further, one can track pairwise connections over frequency bands in an analogous manner to the above as differences between the $j-k$th element and $k-j$th elements.

Importantly, the following proposition shows one can always reconstruct time domain network connectedness measures from our frequency-dependent networks.

\begin{prop}[Reconstruction of Dynamic Network Connectedness]
	\label{prop:4}
	Denote by $d_s$ an interval on the real line from the set of intervals $D$ that form a partition of the interval $(-\pi, \pi)$, such that $\cap_{d_s \in D} d_s = \emptyset, $ and $\cup_{d_s \in D} d_s = (-\pi, \pi)$. We then have that 

	\begin{align}
	\begin{split}
	\mC(u) &= \sum_{d_s \in D} \mC(u,d_s) \\
	\mC_{j\leftarrow\bullet}(u) &= \sum_{d_s \in D} \mC_{j\leftarrow\bullet}(u,d_s) \\
	\mC_{j\rightarrow \bullet}(u)  &= \sum_{d_s \in D} \mC_{j\rightarrow \bullet}(u,d_s) 
	\end{split}
	\end{align}
    where $\mathcal{C}(u)$ are local network connectedness measures aggregated over frequencies.
\end{prop}
\begin{proof}
	See Appendix A.
\end{proof}

In light of the above, all local frequency connectedness measures $\mC(u,d)$ for $u=t/T$ approximate well the time-varying frequency connectedness of the process $\bX_{t,T}$.

\subsection{Obtaining Transitory and Persistent Network Measures}\label{get_measures}

In light of the assumptions that underpin our measures, we conjecture that the economy follows a stable time-varying parameter heteroskedastic VAR (TVP-VAR) model as in (\ref{eq:VAR1}). We follow \cite{petrova2019quasi} who establishes a Quasi Bayesian Local-Likelihood approach for inference in the presence of time-varying parameters. 

For consistent estimation under the QBLL approach, let $\bX_{t,T}$ be a time-series we observe with log probability density $l_t\left(\bX_{t,T}|\bX_{t-1,T},\tilde{\mathbf{\Phi}}(t/T)\right), \: \tilde{\mathbf{\Phi}}(t/T)$ stacks the time-varying autoregressive coefficient matrices into a finite-dimensional vector that satisfies one of the following conditions. 

\begin{enumerate}
\item[i)] $\tilde{\mathbf{\Phi}}_t=\tilde{\mathbf{\Phi}}(t/T)$ is a deterministic process  where $\mathbf{\Phi}(.)$ is a piecewise differentiable function. 
\item[ii)] $\tilde{\mathbf{\Phi}}(t/T)$ is a stochastic process satisfying: $\sup_{j:|j-t|\leq h}||\tilde{\mathbf{\Phi}}_t-\tilde{\mathbf{\Phi}}_j ||^2 =O_p(h/t)$ for $1\leq h\leq t,\: t\rightarrow \infty$.
\end{enumerate}

Both of the above indicate that the parameter sequence drafts gradually over time. The first condition is standard of \cite{dahlhaus2000likelihood} for locally stationary processes which requires the parameter process is a piecewise smooth deterministic function; thus allowing for breaks in parameters. The second condition is a generalization of the first to include stochastic parameter processes exhibiting degrees of persistence necessary for consistent estimation of stochastic driven time-variation. Such condition includes bounded random walk processes and some fractionally integrated processes. The parameters may feature any combination of deterministic trends and/or breaks satisfying conditions i) and ii) above. Our data generating processes (DGPs) in Section \ref{sec:Sim} are examples of such DGPs.\footnote{Figure 1 of \cite{petrova2019quasi} provides figures of examples and provides further discussion around conditions the parameter sequences must satisfy for consistent estimation.}

To obtain the time-varying coefficient estimates at a fixed time point $u=t_{0}/T$, $\widehat{\bPhi}_{1}(u),...,\widehat{\bPhi}_{p}(u)$, and the time-varying covariance matrices, $\widehat{\bSigma}(u)$, we follow the QBLL approach of \cite{petrova2019quasi}. Specifically, this approach uses a kernel weighting function that provides larger weights to observations that surround the period whose coefficient and covariance matrices are of interest. Using conjugate priors, the (quasi) posterior distribution of the parameters of the model are analytical. This alleviates the need to use a Markov Chain Monte Carlo (MCMC) simulation algorithm and permits the use of parallel computing. Note also that in using (quasi) Bayesian estimation  methods, we obtain a distribution of parameters that we use to construct network measures that provide confidence bands. Details of the model and estimation algorithm are in Appendix B.\footnote{Unlike traditional TVP VARs time-variation evolves in a non-parametric manner thus making no assumption on the laws of motion within the model. Typically, the model of \cite{primiceri2005time}, and many extensions, assume parameters evolve as random walks or autoregressive processes.} We provide a computationally efficient package \texttt{DynamicNets.jl} in \textsf{JULIA} and \texttt{DynamicNets} in \textsf{MATLAB} that allows one to obtain our measures on data the researcher desires.\footnote{The packages are available at \url{https://github.com/barunik/DynamicNets.jl} and \url{https://github.com/ellington/DynamicNets}}

To estimate the elements of dynamic adjacency matrix, we first need to truncate the infinite VMA($\infty$) representation of the approximating model with a choice of finite horizon $H$. Here we note that in the frequency domain quantities, $H$ serves only as an approximation factor, and it has no interpretation as in the time domain. Hence in the applications we advise setting the $H$ sufficiently high to obtain a better approximation, particularly when lower frequencies are of interest. We obtain horizon specific measures using Fourier transforms and set our truncation horizon $H=100$. Note, we run all results in this paper for $H\in \{50,100,200\}$, they are qualitatively similar and available upon request. 

Next, estimating dynamic network measures requires the user to choose a kernel and its bandwidth. Typically the larger the bandwidth, the smoother time-evolution of our frequency specific network measures. Therefore, prior to tracking dynamic network connections, it is important the user considers the time-series properties of their data. For example if common peaks (troughs) in the time-series occur frequently and are transient, then a shorter bandwidth may be necessary. Conversely, if tracking network connections among data that evolves gradually over time, like interest rates, a larger bandwidth may be more appropriate. In the context of our study, we use a Normal kernel and explore the implication of bandwidth choice for a variety of data generating processes (DGPs) in Section \ref{Simulations}.

It is noteworthy to mention that the choice of a two-sided kernel can come at a cost; especially if one wishes to use network measures for forecasting purposes. In these cases one may wish to: i) estimate the dynamic network recursively throughout time such that the Normal kernel truncates to use only past values at the time $T$ estimate; or ii) choose a one-sided kernel such as those in \cite{hahn2001identification,barigozzi2020time}. In practice, we encourage researchers to experiment with a variety of bandwidths to ensure results are not driven by its selection. We also encourage authors to use reasonable bandwidths given the data for their application. For example, if one was using these measures for forecasting daily stock return volatilities, a researcher might consider combination forecasts using Bayesian Model averaging to trade on. Alternatively one might look to use variance minimizing kernels  in an attempt to reduce uncertainty around the forecast.\footnote{In the context of our empirical application below where we look at daily realized volatilities of stock returns, we use a bandwidth equal to 8. We also estimate the models using bandwidths of 12, 18, and $\sqrt{T}= \sqrt{3278}\approx 57$. Increasing the bandwidth smooths our network connectedness measures because it assigns larger weights to more distant observations. }

We estimate the $j,k$ element of our dynamic adjacency matrix at time $u=t_0/T$ and horizon $d = (a,b): a,b\in(-\pi,\pi)$ and $a<b$ such that it corresponds to the transitory (high frequency band) and persistent (low frequency band) element of the adjacency matrix respectively as:
\begin{equation} \label{eq:net_estimation}
\Big[ \widehat \btheta(u,d) \Big]_{j,k} = \frac{\widehat{\sigma}_{kk}^{-1} \displaystyle \sum_{\omega \in d} \left( 
	\bigg[ \widehat\bPsi(u,\omega) \widehat \bSigma(u) \bigg]_{j,k} \right)^2 }{ \displaystyle \sum_{\omega \in (-\pi,\pi)} \Bigg[ \widehat\bPsi(u,\omega) \widehat \bSigma(u) \widehat\bPsi^{\top}(u,\omega)  \Bigg]_{j,j} },\end{equation}
where $\widehat\bPsi(u,\omega) = \sum_{h=0}^{H-1} \sum_h \widehat \bPsi(u,h) e^{-i\omega h}$ is an estimate of the impulse transfer function from Fourier frequencies $\omega \in \{ aH / 2 \pi,\ldots, bH/2 \pi \}$ of impulse response functions that cover a specific horizon.\footnote{Note that $i=\sqrt{-1}$.} From this, estimates of Equations (\ref{network})--(\ref{to}) directly follow. For example if the application uses daily data, one may define transitory (short-term) as horizons corresponding to 1--5 days and persistent (long-term) as horizons corresponding to horizons greater than 5 days. This would require defining the band as $(a,b)=(2\pi/5,2\pi)$ for the transitory and $(a,b)=(0,2\pi/5)$ for the persistent networks.

\subsection{Testing for Statistical Differences in Connectedness}\label{test}

We now consider how one can determine, from a statistical perspective, differences between connectedness. We discuss in detail here how one can test for differences in connectedness one computes over different frequency bands. 

In a Bayesian setting there are three alternatives for hypothesis testing. The first is the Bayes factor, the second uses posterior credible intervals, and the third follows statistical decision theory. We follow the latter and utilize the work of \cite{li2014new}, \cite{li2015bayesian} and \cite{liu2020posterior}.\footnote{Bayes factors involve comparing the marginal likelihoods of two competing models and extensively appear in the literature \citep[see e.g.][]{koop2010dynamic,chan2020large}. This is not appropriate in our setting because the network connectedness measures come from a manipulation of a sequence of posterior parameters from the same model; we have no alternative model to specify the marginal likelihood. Using posterior credible intervals is possible and something we consider in the spirit of \cite{cogley2010inflation}. In particular we use the joint posterior distribution of network connectedness measures to compute the probability that connectedness across one frequency band is larger than an analogous measure across another frequency band. These results are in Appendix C.} These studies focus on developing test statistics of a point null hypothesis using the posterior distribution of parameters from a Bayesian model. The approach requires only the posterior distribution of parameters and has various advantages. First, they overcome the problem of the Jeffreys-Lindely paradox. Second, are not sensitive to the prior and are pivotal quantities. Third, they are easy to compute. Crucially, these statistics directly come from quadratic loss functions, as with classical test statistics, and therefore possess the same distributions as their frequentist counterparts. 

\cite{li2015bayesian} develop a Bayesian Lagrange-Multiplier (LM) type test that is asymptotically equivalent to a classical LM test. Using similar assumptions, \cite{liu2020posterior} develop a Bayesian Wald type test that is asymptotically equivalent to a classical Wald test and requires only the posterior mean and posterior variance of parameters under the null hypothesis.\footnote{They also show asymptotic equivalence between their Wald-type test the LM type test in \cite{li2015bayesian}.} Noting that the VMA($\infty$) representation of the VAR are nothing more than a transformation of the VAR parameters, we follow the assumptions in \cite{liu2020posterior} and therefore establish a Bayesian Wald type test for differences between network connectedness across frequency bands. We emphasize our network connectedness measures are manipulations of the VAR parameters themselves and possess a posterior distribution.\footnote{\cite{lutkepohl1990asymptotic} provide the asymptotic distribution for impulse response functions from conventional VAR models one estimates using OLS. \cite{petrova2019quasi} } 

We test the null hypothesis, H$_{0}:\: \mC(u,d_a)=\mC(u,d_b)$ that network connectedness across frequency band $d_a$ and $d_b$ are equivalent against the alternative, H$_{1}:\: \mC(u,d_a)\neq\mC(u,d_b)$. This is equivalent to testing H$_{0}:\: \mC(u,d_a)-\mC(u,d_b)=0$ against H$_{1}:\: \mC(u,d_a)-\mC(u,d_b)\neq0$. The following proposition establishes the test statistic and its asymptotic distribution under the null. We outline the regularity conditions in Appendix A along with a brief discussion on the importance of such conditions.

\begin{prop}[Testing for Heterogeneities in Network Connectedness]
	\label{prop:test}
	Let $\bar{\mathcal{D}}(u)$ and $\mathbf{V}_{\mathcal{D}} \left(\bar{\mathcal{D}}(u)\right)$ denote the time $u$ posterior mean and variance of the difference between network connectedness across frequency band $d_a$ and $d_b$. Then the test statistic under H$_{0}$ 
\begin{eqnarray}\label{eq:test}
\mathbf{W}\left(\bX,\mathcal{D}_{0}(u)\right) &=& q_{\mathcal{D}} + \left(\bar{\mathcal{D}}(u)-\mathcal{D}_{0}(u) \right)^{\top}\left[ \mathbf{V}_{\mathcal{D}\mathcal{D}} \left(\bar{\mathcal{D}}(u)\right)\right]^{-1} \left(\bar{\mathcal{D}}(u)-\mathcal{D}_{0}(u) \right)\\
&=& q_{\mathcal{D}} + \mathbf{Wald}
\end{eqnarray}
where $\mathbf{Wald}= \left(\bar{\mathcal{D}}(u)-\mathcal{D}_{0}(u) \right)^{\top}\left[ \mathbf{V}_{\mathcal{D}\mathcal{D}} \left(\bar{\mathcal{D}}(u)\right)\right]^{-1} \left(\bar{\mathcal{D}}(u)-\mathcal{D}_{0}(u) \right)$, $q_{\mathcal{D}}$ is the number of restrictions, and $\mathcal{D}_{0}(u)=\mC(u,d_a)-\mC(u,d_b)=0$.
\begin{eqnarray*}
\mathbf{W}\left(\bX,\mathcal{D}_{0}(u)\right) - q_{\mathcal{D}} = \mathbf{Wald} + O_p(1) \rightarrow^{d} \chi^2\left(q_{\mathcal{D}}\right)
\end{eqnarray*}
\end{prop}
\begin{proof}
  See Appendix A.
\end{proof}

It is important to note that it is straightforward to generalize Equation (\ref{eq:test}) to include multiple restrictions. This may be applicable if one requires testing equivalence among more than two frequency bands. We also note that one may also wish to utilize the above to test for differences between directional network connections over frequency bands. In the context of the above one would take the difference between net-directional connections, or pairwise directional connections, over frequency bands and compute the test in an analogous manner to below. 

For estimation purposes, the test statistic only requires the posterior mean and the posterior variance of $\mathcal{D}(u)$, $\bar{\mathcal{D}}(u)$ and $\mathbf{V}\left(\bar{\mathcal{D}}\right)$. Let $\left\lbrace\mathcal{D}^{[r]}\right\rbrace^{R}_{r=1}$ denote the posterior draws such that $\mathcal{D}^{[r]}=\widehat{\mC}^{[r]}(u,d_a)-\widehat{\mC}^{[r]}(u,d_b)$ is the $r$th posterior draw of the difference between estimates of network connectedness over frequency band $d_a$ and $d_b$. Then, the estimate of our test statistic for heterogeneities between network connectedness is given by

\begin{eqnarray}\label{eq:test_est}
\widehat{\mathbf{W}}\left(\bX,\mathcal{D}(u)=0\right) = \frac{\frac{1}{R}\sum^{R}_{r=1}\left(\mathcal{D}^{[r]}(u)\right)^2}{\frac{1}{R}\sum^{R}_{r=1}\left(\mathcal{D}^{[r]}(u)-\bar{\bar{\mathcal{D}}}(u) \right)^2},\: \text{with} \: \bar{\bar{\mathcal{D}}}(u)=\frac{1}{R}\sum^{R}_{r=1}\mathcal{D}^{[r]}(u) 
\end{eqnarray}

Under the null, we have $\mathbf{W}\left(\bX,\mathcal{D}_{0}(u)\right) - q_{\mathcal{D}}  \rightarrow^{d} \chi^2\left(q_{\mathcal{D}}\right)$ with $q_{\mathcal{D}}=1$ in this particular case.\footnote{ In the classical setting one uses Wald tests to check a variety of restrictions, such as a parameter of interest being equal to zero, or equivalence between two parameters of interest under the null hypothesis. When one tests the latter, the resulting Wald test is $\backsim \chi^2(1)$.} Thus, we only need to compare $\widehat{\mathbf{W}}\left(\bX,\mathcal{D}(u)=0\right)$ to the critical values of the $\chi^2\left(1\right)$ distribution. Rejecting the null implies that the time $u$ network connectedness over frequency band $d_a$ is statistically different to the corresponding time $u$ network connectedness over frequency band $d_b$. 

It may also be pertinent to test for differences in connectedness over time. Here we fix the frequency band, and now consider differences over time. In this case we test the null hypothesis, H$_{0}:\: \mC(u_1,d)=\mC(u_2,d)$ that network connectedness at time $u_1$ and $u_2$ across frequency band $d$ are equivalent against the alternative, H$_{1}:\: \mC(u_1,d)\neq\mC(u_2,d)$. This is equivalent to testing H$_{0}:\: \mC(u_1,d)-\mC(u_2,d)=0$ against H$_{1}:\: \mC(u_1,d)-\mC(u_2,d)\neq0$. Now letting $\bar{\mathcal{D}}(s)$ and $\mathbf{V}_{\mathcal{D}\mathcal{D}} \left(\bar{\mathcal{D}}(s)\right)$ denote the posterior mean and variance of the difference between network connectedness at times $u_1$ and $u_2$ and replacing $\bar{\mathcal{D}}(u)$ and $\mathbf{V}_{\mathcal{D}\mathcal{D}} \left(\bar{\mathcal{D}}(u)\right)$ in Proposition \ref{prop:test} with these quantities delivers a Wald-type test statistic following a $\chi^2\left(1\right)$ distribution. Again, all we need to do is compare the test statistic with critical values of the $\chi^2\left(1\right)$ distribution.

\section{Monte Carlo Study}\label{Simulations}\label{sec:Sim} 
In this section, we conduct a Monte Carlo exercise to understand the finite sample properties of our connectedness measures. In order to motivate the need to focus on  connections forming conditional on the persistence of shocks, we generate data with different levels of persistence throughout time and also changes in the covariance structure. This will induce differences in network connectedness measures we compute over different frequency bands. Here we concentrate on low and high frequency bands and consider four different data generating processes (DGP) to highlight their uses.

For simplicity, we focus on bivariate VAR(2) models with time-varying parameters and time-varying covariance matrices: 
\begin{eqnarray*}
\bX_{t,T} &=& \bPhi_{0}(u) + \bPhi_{1}(u)\bX_{t-1,T} + \bPhi_{2}(u)\bX_{t-2,T} + \bepsilon_{t,T}, \\ \bepsilon_{t,T} &=& \bSigma^{-1/2}(u)\bbeta_{t,T}, \: \bbeta_{t,T} \backsim \left(0,\mathbf{I}_{2}\right)
\end{eqnarray*} 

where $\bPhi_{0}(u)$ contains the time-varying intercepts and $\bPhi_{1}(u)$ and $\bPhi_{2}(u)$ contain the time-varying autoregressive parameters. The time-varying covariance matrix $\bSigma(u) = \mathbf{A}^{-1}(u)\mathbf{H}(u)\left(\mathbf{A}^{-1}(u)\right)^{\top}$ with $\mathbf{A}^{-1}(u)$ being a lower triangular matrix with a unit diagonal and $\mathbf{H}(u)$ is a $2 \times 2$ diagonal matrix.

\textbf{DGPI:} Our first DGP has residuals such that, $ \bbeta_{t,T} \backsim \text{NID}\left(0,\mathbf{I}_{2}\right)$. The time-varying intercepts follow the process:
\begin{eqnarray*}
\left[\bPhi_{0}(u)\right]_j &=& 0.0025 \sin \left(0.004\pi t\right) + 0.15 \sum^{t}_{i=1} \frac{\nu_{i}}{\sqrt{t}},  \: \nu_i \backsim \text{NID} \left(0,0.001^2 \right),\: j=1,2 \\
\end{eqnarray*}

For the time-varying autoregressive parameters, we have
\begin{eqnarray*}
\left[\bPhi_{g}(u)\right]_{j,k} = \begin{cases}
0.05 \sin \left(0.002\pi t\right) + 0.75 \sum^{t}_{i=1} \frac{\kappa_{i}}{\sqrt{t}},\: &  t \in \{1,...,500\} \: \forall g,j,k=1,2 \\
0.45 \sin \left(0.002\pi t\right) + 0.75 \sum^{t}_{i=1} \frac{\kappa_{i}}{\sqrt{t}},\: & t \in \{501,...,1000\} \: j=k=1,\: j=k=2 \\
0.05 \sin \left(0.002\pi t\right) + 0.75 \sum^{t}_{i=1} \frac{\kappa_{i}}{\sqrt{t}},\: &  t \in \{501,...,1000\} \: j=1,k=2\:\text{\&}\: j=2,k=1  
\end{cases}
\end{eqnarray*}

with $\kappa_i \backsim \text{NID} \left(0,0.0001^2 \right)$. The (2,1) element of $\mathbf{A}(u)$ have the following dynamics:
\begin{eqnarray*}
\left[\mathbf{A}(u)\right]_{2,1}= \begin{cases} 0.03 \sin \left(0.002\pi t\right) + 0.7 \sum^{t}_{i=1} \frac{\upsilon_{i}}{\sqrt{t}},\: &  t \in \{1,...,500\}\\
1.5 \sin \left(0.002\pi t\right) + 0.7 \sum^{t}_{i=1} \frac{\upsilon_{i}}{\sqrt{t}},\: &  t \in \{501,...,1000\}
\end{cases}
\end{eqnarray*}
with $ \upsilon_i \backsim \text{NID} \left(0,0.3^2 \right)$. The diagonal elements of $\mathbf{H}(u)$ follow
\begin{eqnarray*}
\log \left[\mathbf{H}(u)\right]_{j,j} &=& \mu_{j} + \lambda_{j}\left( \log \left[\mathbf{H}(u-1)\right]_{j,j} - \mu_{j} \right) + \xi_{j,t}
\end{eqnarray*}

where $ \xi_{j,t} \backsim \text{NID}\left(\mu_{j}, 0.1^2/(1-\lambda_{j}) \right),\: \mu_j=0.01,\: \lambda_j=0.95$.

This DGP has little to no dependence for the first 500 observations which means connectedness at both high and low frequency bands will be low and close to zero. The latter half of the sample sees the AR coefficients in each equation become persistent as the sin wave becomes negative. Note also that the contemporaneous relationship intensifies. This induces connections at high frequency bands while connections at low frequency bands should be low and close to zero. Note we also allow for non-Gaussian residuals in this DGP such that we draw $\bbeta_{t,T}$ from a multivariate student-$t$ distribution with 5 degrees of freedom. These results are in Appendix C.

\textbf{DGPII:} For our second DGP, the time-varying intercepts follow the process:
\begin{eqnarray*}
\left[\bPhi_{0}(u)\right]_j &=& 0.25 \sin \left(0.004\pi t\right) + 0.15 \sum^{t}_{i=1} \frac{\nu_{i}}{\sqrt{t}},  \: \nu_i \backsim \text{NID} \left(0,0.1^2 \right),\: j=1,2 \\
\end{eqnarray*}

and the time-varying autoregressive parameters follow:
\begin{eqnarray*}
\left[\bPhi_{g}(u)\right]_{j,k} = 0.25 \sin \left(0.004\pi t\right) + 0.75 \sum^{t}_{i=1} \frac{\kappa_{i}}{\sqrt{t}}, \forall g,j,k=1,2
\end{eqnarray*}

with $\kappa_i \backsim \text{NID} \left(0,0.3^2 \right)$. The (2,1) element of $\mathbf{A}(u)$ and the diagonal elements of $\mathbf{H}(u)$ follow the processes:
\begin{eqnarray*}
\left[\mathbf{A}(u)\right]_{2,1}&=& 0.3 \sin \left(0.008\pi t\right) + 0.7 \sum^{t}_{i=1} \frac{\upsilon_{i}}{\sqrt{t}} \\
\log \left[\mathbf{H}(u)\right]_{j,j} &=& \mu_{j} + \lambda_{j}\left( \log \left[\mathbf{H}(u-1)\right]_{j,j} - \mu_{j} \right) + \xi_{j,t},
\end{eqnarray*}

where $\upsilon_i \backsim \text{NID} \left(0,0.3^2 \right)$ and $\xi_{j,t} \backsim \text{NID}\left(\mu_{j}, 0.1^2/(1-\lambda_{j}) \right),\: \mu_j=0.01,\: \lambda_j=0.95 $. 

This DGP induces two distinct periods of persistence during observations 100-200 and 600-700 which amplifies connectedness at the low frequency band.

\textbf{DGPIII:} Our third DGP is the same as DGPII but relaxes the assumption that $\bbeta_{t,T}$ are Gaussian. Instead we assume that the residuals follow a multivariate student-$t$ distribution with $5$ degrees of freedom.

\textbf{DGPIV:} Our fourth DGP is the same as DGPII, but increases the periodicity of the sin functions in the time-varying autoregressive matrices from $\sin \left(0.004\pi t\right)$ to $\sin \left(0.006\pi t\right)$. This generates three distinct periods of persistence during observations 50-150, 300-450, and 700-850. 

\begin{figure}[!hp]
\centering
\scalebox{1.00}{\includegraphics{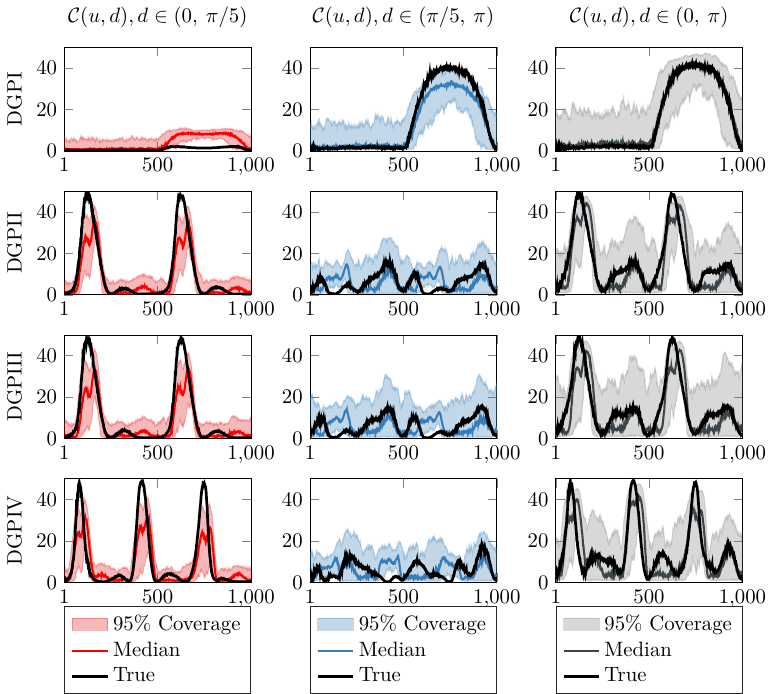}}
    \caption{\textbf{Dynamic network connectedness measures: True and fitted values} \\ \small{Notes: This figure plots the true network connectedness measures for three data generating processes following bi-variate TVP VAR(2) models along with the median and 95\% quantiles of estimated network connectedness measures with bandwidth $W$=8. The left columns report network connectedness on the low-frequency band, $d\in(0,\:\pi/5)$, the middle columns show network connectedness on the high-frequency band, $d\in(\pi/5,\: \pi)$, and the right columns show the aggregate network connectedness such that $d\in(0,\:\pi)$. DGPI (top row) is a TVP VAR(2) model with Gaussian errors, we introduce a break in the time-varying autoregressive matrices and contemporaneous relations from observation 500 that induces large connections across the high frequency band. DGPII (second row) is a TVP VAR(2) where time-varying intercepts and autoregressive matrices following sin wave functions with a stochastic error, time-varying covariance matrix where the off-diagonals follow sin wave functions with a stochastic error, and the diagonal elements follow a stationary AR(1) processes. DGPIII (third row) is a TVP VAR(2) model with student-$t$ errors, time-varying intercepts and autoregressive matrices following sin wave functions with a stochastic error, time-varying covariance matrix where the off-diagonals follow sin wave functions with a stochastic error, and the diagonal elements follow a stationary AR(1) processes. DGPIV (bottom row) is the same as DGPII, but with an increase in the periodicity of the respective sin wave functions the time-varying intercepts and autoregressive matrices follow.}}
      \label{fig:MC_W8}
\end{figure}

For each of the four DGPs, we generate 100 simulations of length $T=1000$ and compute the  network connectedness measures. We use the median over these simulations as the true network connectedness. Then, for each of the 100 simulations of DGPI--DGPIV, we fit the TVP-VAR model we outline in Section \ref{get_measures}. In fitting this model we take 1000 draws from the posterior distribution, calculate our network connectedness measures and then save the posterior median. For this exercise, we compute network connectedness on two frequency bands that cover the spectrum. The low-frequency band, which empirically pertains to persistent network connections, is $d\in(0,\:\pi/5)$, and the high-frequency band, pertaining to transitory network connections, is $d\in(\pi/5,\: \pi)$. For completeness, we compute the aggregate connectedness measures that considers the entire spectrum such that $d\in(0,\:\pi)$; this corresponds to a dynamic version of the \cite{diebold2014} connectedness measure. 

A final noteworthy point is the choice of bandwidth, $W$, for the kernel weights that induce time-variation into the TVP-VAR model. In this exercise, we consider $W=\{8,\:12,\:18\}$. However, for ease of exposition, we only report plots of the network connectedness measure estimates using $W$=8 in the main text, results of fitted values from $W=\{12,\:18\}$ are in Appendix C. The larger the bandwidth, the smoother the network measures become. This is because larger bandwidths assign weights to a higher number of observations around the one of interest. In general, we find that larger $W$ results in poorer fit. This highlights the importance of selecting an appropriate bandwidth for the kernel weights relative to the data application. 

As we discuss in \ref{get_measures}, from a practical perspective, we encourage researchers to explore the robustness of their results to different bandwidths. For example, consider a low frequency forecaster looking to predict returns one-month ahead today, using our network connectedness measures, would likely place zero weight on data from the burst of the dot-com bubble and 2008 recession. Likewise, a high frequency investor would likely place little to no weight on data from one-year prior. As we show in Appendix C, lengthening the bandwidth in our Monte Carlo experiment causes the surges in estimates of connectedness to be more gradual. 

If the data is low frequency data such as monthly yields, then one could argue to use a wider bandwidth as changes in these data are far smoother than returns or return volatility. If one is looking to describe the nature of connections that form on persistent and transitory components of shocks then we encourage researchers to use multiple bandwidths to check how these dynamics are influenced by such changes. We do not suspect these changes would drastically change the conclusions or results in most applications. 

In Appendix C, we conduct further robustness checks for our simulation analysis. In Appendix C.3 we provide analysis on the performance of rolling VAR models for our DGPs. These results show that such connectedness estimates less accurate relative to our approach. In particular, estimates are highly sensitive to the window size and fail to accurately capture the peaks and troughs in connectedness; the latter is prominent as the complexity of the DGP increases (e.g. DGPII-DGPIV). Meanwhile in Appendix C.4 we conduct a Monte Carlo study on larger scale VAR models under DGPI containing $N$=10 and $N$=25 variables respectively. These results show that our approach is robust to increasing the number of variables and tracks connections well within larger-scale models.

\begin{figure}[!hp]
\centering
\scalebox{1.00}{\includegraphics{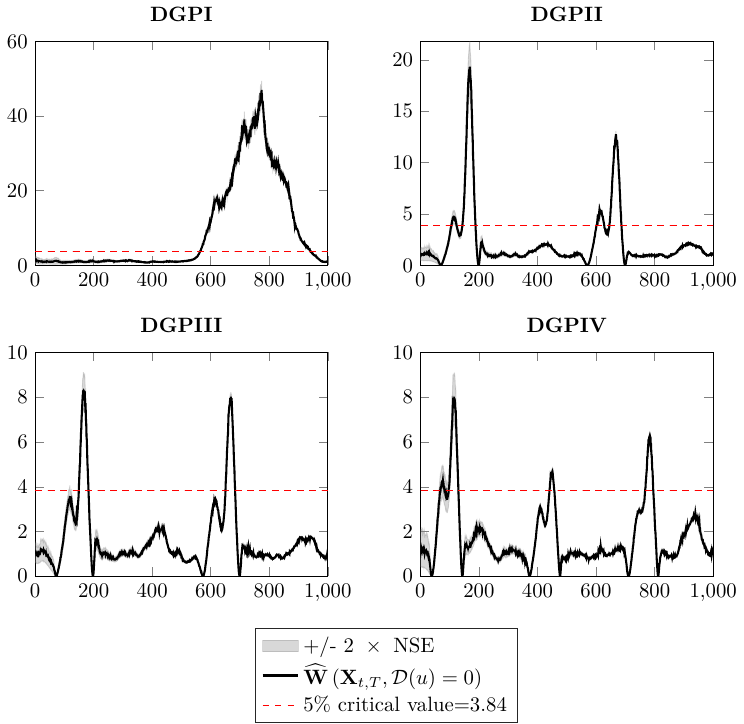}}
    \caption{\textbf{Tests for differences between network connectedness across different frequency bands} \\ \small{Notes: This figure plots the test statistics and $+/-$ 2 $\times$ the numerical standard errors for heterogeneities between network connectedness measures across different frequency bands from four DGPs. The test statistic is $\widehat{\mathbf{W}}\left(\bX_{t,T},\mathcal{D}(u)=0\right)-1$ where $\mathcal{D}(u)=\widehat{\mC}(u,d)-\widehat{\mC}(u,c)$ with $d \in(0,\:\pi/5)$ corresponding to the low frequency band and $c \in(\pi/5,\:\pi)$ corresponding to the high frequency band. The dashed black line is the 5\% critical value from the $\chi^2(1)$ distribution =3.84. Values greater than 3.84 reject the null hypothesis of equivalent network connections across frequency band $d$ and $c$ in favor of differences. DGPI (top left panel) is a TVP VAR(2) model with Gaussian errors, we introduce a break in the time-varying autoregressive matrices and contemporaneous relations from observation 500 that induces large connections across the high frequency band. DGPII (top right panel) is a TVP VAR(2) where time-varying intercepts and autoregressive matrices following sin wave functions with a stochastic error, time-varying covariance matrix where the off-diagonals follow sin wave functions with a stochastic error, and the diagonal elements follow a stationary AR(1) processes. DGPIII (bottom left panel) is a TVP VAR(2) model with student-$t$ errors, time-varying intercepts and autoregressive matrices following sin wave functions with a stochastic error, time-varying covariance matrix where the off-diagonals follow sin wave functions with a stochastic error, and the diagonal elements follow a stationary AR(1) processes. DGPIV (bottom right panel) is the same as DGPII, but with an increase in the periodicity of the respective sin wave functions the time-varying intercepts and autoregressive matrices follow.}}
      \label{fig:test}
\end{figure}

Figure \ref{fig:MC_W8} reports the true network connectedness measures and the median and 95\% quantiles of corresponding estimates from the TVP VAR model using a kernel bandwidth of $W$=8. We report network connectedness over the low-frequency-band, the high-frequency-band, and aggregate, in the left, middle, and right columns respectively. The top row corresponds to DGPI, and the second, third and fourth rows results from DGPII, DGPIII, and DGPIV respectively. As we can see, the distribution of estimates for each DGP track the true values remarkably well. In almost all cases, the true value lies within the 95\% quantiles of the distribution from model estimates. This plot shows that our method provides an accurate representation of horizon specific network connectedness, even when the underlying process has complex dynamics and the true error distribution is non-Gaussian.

Figure \ref{fig:test} further reports the estimates, and numerical standard error bounds of our test for heterogeneities between high frequency band and low frequency band network connectedness measures (Equation \ref{eq:test_est}) for each of the four DGPs we use in the Monte Carlo study. Specifically, for each observation we test the null hypothesis $\mathcal{D}(u)=0$ where $\mathcal{D}(u)=\widehat{\mC}(u,d)-\widehat{\mC}(u,c)$ with $d \in(0,\:\pi/5)$ corresponding to the low frequency band and $c \in(\pi/5,\:\pi)$ corresponding to the high frequency band. In each plot, we also report the 5\% critical value from the $\chi^2(1)$ distribution of 3.84. Test statistics exceeding this value reject the null in favor of heterogeneities between network connectedness measures across frequency bands.\footnote{We obtain numerical standard errors in a similar manner to \cite{li2015bayesian}.} As we can see, our test statistic identifies significant differences between low and high frequency band network connectedness measures that correspond with the peaks we observe in Figure \ref{fig:MC_W8} for each DGP. We can see with the non-Gaussian DGP, DGPIII the estimates of the test statistics are smaller relative to the analogous Gaussian DGP, DGPII. However, there are still clear rejections.

\begin{figure}[!hp]
\centering
\scalebox{1.00}{\includegraphics{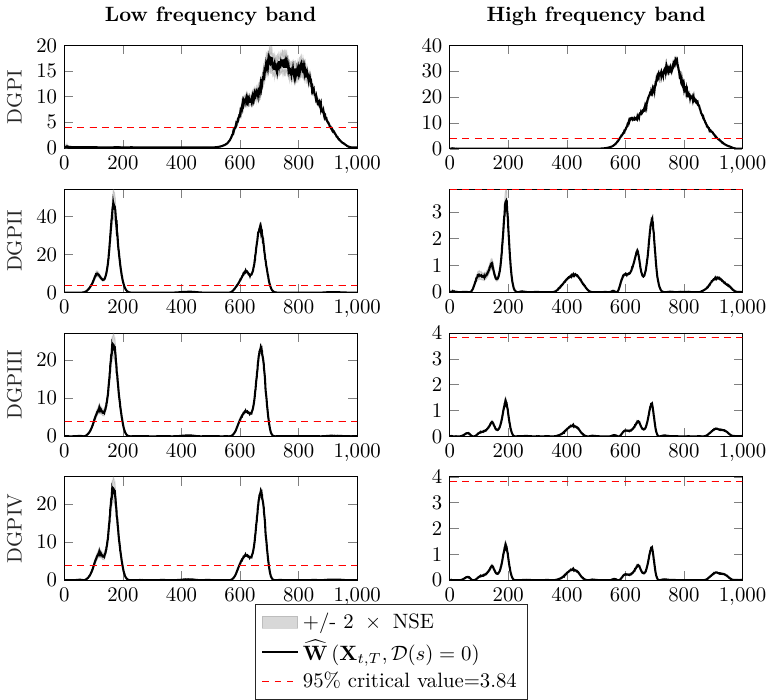}}
    \caption{\textbf{Tests for differences between network connectedness over time} \\ \small{Notes: This figure plots the test statistics and $+/-$ 2 $\times$ the numerical standard errors for heterogeneities between network connectedness measures across different frequency bands from four DGPs. The test statistic is $\widehat{\mathbf{W}}\left(\bX_{t,T},\mathcal{D}(s)=0\right)-1$ where $\mathcal{D}(u)=\widehat{\mC}(u_1,d)-\widehat{\mC}(u_2,d)$ with $u_1=1,\: u_2=\{2,3,\dots,1,000\}$. Here $d \in(0,\:\pi/5)$ corresponds to the low frequency band (LHS plots) and $d \in(\pi/5,\:\pi)$ corresponds to the high frequency band (RHS plots). The dashed black line is the 5\% critical value from the $\chi^2(1)$ distribution =3.84. Values greater than 3.84 reject the null hypothesis of equivalent network connections at time $u_1,\:u_2$ across frequency band $d$. DGPI (top row) is a TVP VAR(2) model with Gaussian errors, we introduce a break in the time-varying autoregressive matrices and contemporaneous relations from observation 500 that induces large connections across the high frequency band. DGPII (second row) is a TVP VAR(2) where time-varying intercepts and autoregressive matrices following sin wave functions with a stochastic error, time-varying covariance matrix where the off-diagonals follow sin wave functions with a stochastic error, and the diagonal elements follow a stationary AR(1) processes. DGPIII (third row) is a TVP VAR(2) model with student-$t$ errors, time-varying intercepts and autoregressive matrices following sin wave functions with a stochastic error, time-varying covariance matrix where the off-diagonals follow sin wave functions with a stochastic error, and the diagonal elements follow a stationary AR(1) processes. DGPIV (fourth row) is the same as DGPII, but with an increase in the periodicity of the respective sin wave functions the time-varying intercepts and autoregressive matrices follow.}}
      \label{fig:test2}
\end{figure}

We now test for differences in connectedness measures across the same frequency band over time. To do so, we test the first time period against all remaining 999 observations from our DGPs, $u_1=1,\: u_2=\{2,3,\dots,1,000\}$. We do this for both the low-frequency band, $d\in(0,\:\pi/5)$, and the high-frequency band, $d\in(\pi/5,\: \pi)$. Figure \ref{fig:test2} reports the estimates of the test statistics, their numerical standard error and the corresponding 95\% critical value. First, considering DGPI, it is clear that connectedness across the high frequency band exhibits significant differences between the first observation when we increase connections at observation 500. It is noteworthy to mention that we also see this for the low frequency band. This is expected as we can see that the estimates from our simulations exhibit slight bias here. However, for DGPII--DGPIV we have rejections relative to the first observation at the corresponding periods where we create connectedness across the low frequency band, and no rejections across the high frequency band.

Overall, this shows that our testing procedure indicates rejections of equality in connectedness forming over different frequency bands, and over time for connectedness across the same frequency band, where we should expect to see such differences.

\section{Monitoring Persistence in Uncertainty Networks using S\&P500 Stocks}\label{sec:empirical}

Changes in uncertainty can play a key role in driving business cycles and financial turmoil \citep{bloom2018really}. Identifying the sources of such risks is a focus for researchers and practitioners. Some quantify systemic risks emanating from financial markets and sectors \citep[e.g.][]{billio2012econometric,acemoglu2015systemic}, while others examine how sectoral shocks affect aggregate fluctuations \citep[e.g.][]{gabaix2011granular,acemoglu2017microeconomic,atalay2017important}. Related to the above, and in response to financial crises, many countries are implementing policies to monitor systemic risk and financial stability. Therefore, we use our framework to identify new measures of transitory and persistent linkages for S\&P500 constituents. In section \ref{sec:sector}, we examine sectoral connectedness as well as network structures at a granular level for financial firms. 

We use high-frequency tick data for all stocks listed on the S\&P500 from 5 July 2005 to 31 August 2018 and compute realised volatility (RV) for all stocks in the sample. To obtain firm-level RVs, we restrict our analysis to five-minute returns during New York Stock Exchange (NYSE) trading hours (i.e. 09:30-16:00). The data are time-synchronised using the same timestamps, eliminating transactions executed on Saturdays and Sundays, US holidays, 24-26 December and 31 December to 2 January due to low activity on these days. This leaves us with 3278 trading days. After cleaning the data, we are left with a cross section of 496 stocks.

To obtain our network connectivity measures, we estimate a TVP VAR model on $N$=496 stocks with $p$=2 lags on our $T$=3278 days of data. We estimate our horizon-specific dynamic network measures on a 48-core server. For each $t\in\{1,2,\dots, T\}$, we generate 500 simulations of the (quasi) posterior distribution, resulting in a total estimation time of 10 days. We define transitory (short-term) network links as those that form over a 1-5 day horizon, and persistent (long-term) network links at horizons greater than 5 days (i.e. 5 days to the $\infty$ horizon). Our choice of these horizons stems from the existing literature on volatility modelling using high-frequency data, which shows that daily and weekly fluctuations contain salient information for future volatility \citep[e.g.][]{corsi2012discrete}.\footnote{In this case, due to the size of the system, we diagonalise the covariance matrix of the VAR as an additional precaution to avoid overfitting. For the sectoral network measures in section \ref{sec:sector}, we use a full covariance matrix.} 

\begin{figure}[!ht]
\centering
\scalebox{0.7}{\includegraphics{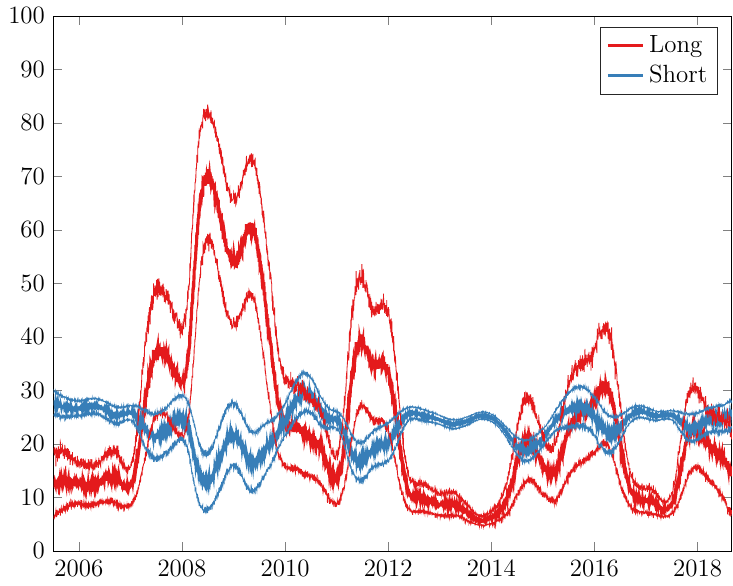}}
		\caption{\textbf{Horizon specific dynamic total network connectedness for S\&P500 constituents} \\ \small{This figure plots the posterior median and 1-standard deviation percentiles of horizon specific dynamic total network connectedness, $\mC(u,d),\: d \in \{\text{S},\text{L}\} $ from July 5, 2005 to August 31, 2018. S refers to the short-term,  transitory connectedness, which we define as 1 day to 1 week; and L refers to long-term, persistent connectedness, which we define as horizons $>$ 1 week. The spectrum with which the horizons stem from link to the frequency with which we observe the data.}}
      \label{NETSSP}
\end{figure}

Figure \ref{NETSSP} plots measures of transitory and persistent network connectedness from 5 July 2005 to 31 August 2018. Overall, there are significant differences in the level of horizon-specific connectedness across our estimation sample. In general, long-term linkages are muted during periods of economic/financial calm. However, it is clear that long-term connectedness spikes during periods of economic recession or major stock market events. For example, long-term connectedness starts to rise in 2006 and continues to rise during the 2007-2009 recession. In addition, we can see that long-term connectedness increases during 2010-2012. This may be due to fears of contagion from the European sovereign debt crisis, the 2010 flash crash and when the S\&P500 entered a bear market in 2011, albeit a short-lived one. We can also see an increase in short- and long-term connectedness in mid-to-late 2015, which is consistent with the stock market sell-off starting in August 2015; this may also be related to fears of contagion from the Chinese stock market crash in late 2015. Overall, we see that long-term network connectedness increases during periods of high systemic risk across our sample.

\subsection{Transitory and Persistent Network Connectedness of S\&P500 Sectors}\label{sec:sector}

Here we focus on the overall network connectedness driven by transitory and persistent shocks to companies in a given sector. We classify stocks into eleven main sectors according to the Global Industry Classification Standard (GICS)\footnote{GICS is an industry taxonomy developed by MSCI and Standard \& Poor's for use by the global financial community.}. These are: Consumer Discretionary (COND) with 73 stocks; Consumer Staples (CONS) with 34 stocks; Health Care (HLTH) with 53 stocks; Industrials (INDU) with 73 stocks; Information Technology (INFT) with 67 stocks; Materials (MATR) with 33 stocks; Real Estate (REAS) with 29 stocks; Financials (SPF) with 66 stocks; Energy (SPN) with 36 stocks; Communication Services (TELS) with 6 stocks; and Utilities (UTIL) with 26 stocks. Further details, including descriptive statistics (Table D.1 of Appendix D) for annualised daily RVs, which we compute as $100\times\sqrt{252 \times RV_t}$ pooling information across companies within each sector over the sample period from 5 July 2005 to 31 August 2018, are reported in Appendix D. The energy sector has the highest mean, while the financial sector has the highest standard deviation, skewness and kurtosis. Overall, we can see that these sectoral RVs show significant differences in terms of the first four moments as well as the minimum and maximum values.  

For each of the 11 sectors, we obtain dynamic network measures by estimating a TVP VAR model on all stocks with two lags on our 3278 days of data. On each day, we take 500 draws from the (quasi) posterior distribution.\footnote{We estimate our horizon-specific dynamic network measures on a 64-core server, resulting in a total estimation time of about 4--5 hours to obtain network estimates from the 11 sectors.}. We define transient and persistent network connectivity in the same way as above.

In Figure \ref{fig:app_net_sec} we plot the posterior median and 95\% confidence bands for transitory and persistent network connectivity for each sector as in Equation \ref{network}. These follow directly from the manipulations of the estimates of the dynamic adjacency matrices (see equation \ref{eq:net_estimation}). Grey bars in these figures represent periods where there are statistically significant differences at the 5\% level between transient and persistent network connectivity.\footnote{We plot the values of these test statistics in Figure D.1 in Appendix D.} Overall, we observe significant differences between these horizon-specific networks for each sector. In general, network connectedness due to the persistent component of shocks exceeds that due to the transitory component of shocks during periods of market turbulence. Then, during periods of calm, the transitory part of the networks becomes more pronounced. 

Comparing these measures across sectors, we can see from the real estate and financial sectors that surges in network connectedness driven by persistent shocks drive uncertainty in the sector for much longer periods of time during the Great Recession relative to other sectors, e.g. CONS and HLTH. In addition, the magnitude of persistent network connectedness from the real estate and financial sectors is much greater. Note also that throughout this period we observe much more frequent evidence in favour of statistical differences between networks driven by transitory and persistent shocks for these sectors; particularly relative to COND, CONS, HLTH and MATR. Although the other sectors show spikes in persistent network connectedness, these do not occur until around February 2008. This highlights how long-term systemic risks within the real estate and financial sectors intensify during this period and are stronger relative to other sectors.

We also see clear spikes in persistent sectoral network connectedness in May-October 2011 and again in 2015-2016. The former coincides with the S\&P500 entering a bear market and the latter with declines in major stock markets around the world. In 2015-2016, we see much lower long-term connectedness of consumer staples, utilities, real estate and telecoms relative to other sectors. We expect the long-term linkages of industrials, materials, energy, information technology and financials to be strong during this period, as the decline in global equity markets is linked to falling commodity prices and the depreciation of Asian currencies against the US dollar. 

\begin{landscape}
\begin{figure}
\centering
\includegraphics[width=8in]{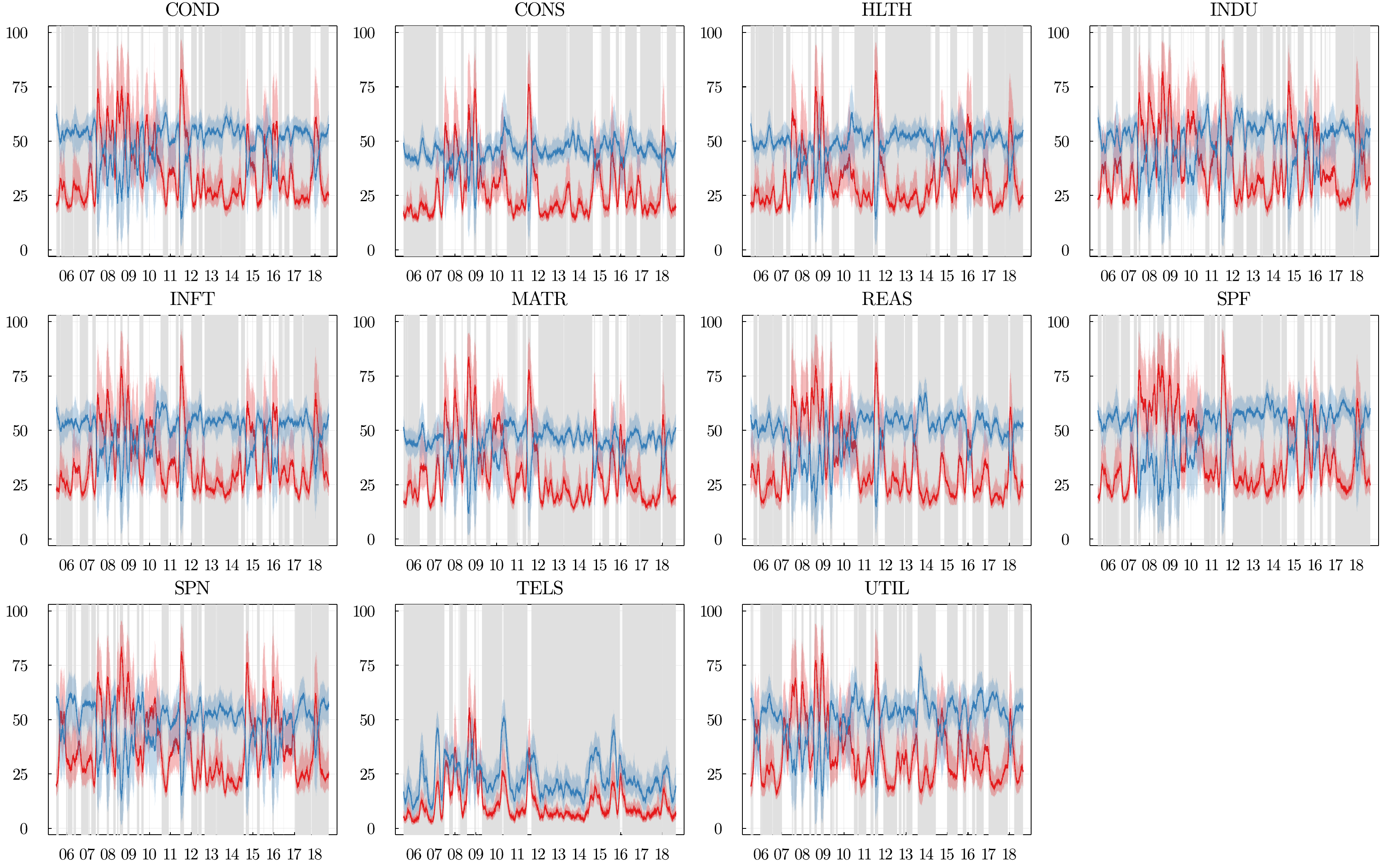}
    \caption{\textbf{Dynamics and Persistence in the U.S. Uncertainty Networks} \\ \small{This figure plots the (quasi) posterior median and 95\% confidence bands of network connectedness specific to the transitory, or short-term (in blue) and persistent, or long-term (in red) shocks to realized volatility of the S\&P500 sectors: Consumers Discretionary (COND), Consumer Staples (CONS), Health Care (HLTH), Industrials (INDU), Information Technology (INFT), Materials (MATR), Real Estate (REAS), Financials (SPF), Energy (SPN), Comunication Services (TELS), and Utilities (UTIL) from July 8, 2005 to August 31, 2018. We define transitory (short-term), as connections made over the 1 day to 1 week horizon; we characterize connections greater than 1 week as persistent (long-term). Grey bars indicate periods with significant heterogeneity in persistence structures.}}
      \label{fig:app_net_sec}
\end{figure}
\end{landscape}

In terms of sectoral network linkages due to transitory shocks, the main differences we observe are in the size of the linkages. Note, however, that there are subtle differences in the time profiles. For each sector, we see that transitory linkages are much stronger during periods of uncertainty at the firm level than persistent linkages during periods of calm. 

Digging deeper into our investigation of whether there are significant differences between transitory and persistent network connectedness for each sector, two main points emerge. First, we document significant heterogeneity in persistence over substantial periods of time. Figure \ref{fig:app_net_sec} shows that during periods of tranquillity, transitory network linkages between sectors are stronger relative to linkages from persistent shocks. During periods of turbulence, however, our results show that persistent linkages intensify in all sectors, but their magnitude differs considerably.\footnote{From Figure D.1 in Appendix D we also document that the magnitude of the differences evolves substantially over time.} Second, statistical differences between transitory and persistent connectedness in one sector do not necessarily imply differences in other sectors. For example, over the 2014-2017 period, we see clear differences in the rejections that transitory and persistent connectedness are equivalent in REAS, SPF, SPN and UTIL.

In general, the temporal nature of our results aligns well with \cite{bianchi2019modeling}, which documents a regime dependent impact of systemic risk on financial markets. They also provide empirical support for \cite{gabaix2011granular} and \cite{acemoglu2012network}, and uncover new measures of sectoral uncertainties (or sector-wide risks). Testing for statistical differences between transitory and persistent sectoral network connectedness adds further substance to our suggestion that one should consider dynamic network structures that form across frequency bands. 

For researchers, our measures of network connectedness may contain useful information for real economy or forecasting purposes; there is already evidence that network connectedness contains predictive content for the real economy \citep[e.g.][]{barunik2020asymmetric}. For practitioners, tracking persistence in sectoral networks can be useful for informing macroprudential policy. This is because one can use these measures as online monitoring tools to study the evolution and persistence of sectoral network connectedness.

Although we provide evidence of substantial heterogeneity in the persistence of network structures, there are commonalities in the time profiles. We attribute this to the high degree of correlation between the data used to proxy uncertainty in our investigation. \cite{herskovic2016common} exploits the correlation structure of idiosyncratic return volatilities and shows that a common factor among the drivers of firm-level volatilities has pricing implications. Our network connectedness measures, by definition, refer to this correlation structure and provide an aggregate description of the network at each point in time. However, our network measures are more informative. We are able to obtain measures that contain information about the overall network structure. These may relate to directional connections (in and out degrees) or concentration (i.e. a high influence of a small number of firms/nodes on the overall network); both of which have been shown to contain information with economic implications \citep[see e.g.][]{herskovic2018networks,herskovic2020firm}. 

\subsection{Network Connections at a Granular Level}

By focusing on shocks to a single financial institution that affect the wider system, our research contributes to the large literature on measuring systemic risk. Many understand systemic risk as many market participants realizing severe losses as a result of propagation throughout the (financial) system.\footnote{For a comprehensive review of the literature on systemic risk, see \cite{benoit2017risks}.} During periods of financial turbulence, uncertainty shocks, the drying up of liquidity, and insolvencies have the ability to spread rapidly affecting many institutions across the market. Lessons from recent financial crises spur the demand for financial regulations in order to mitigate firm behaviors consistent with increasing systemic risk.   

From a prudential perspective, safeguarding against this type of risk requires quantifying systemic risk. The existing literature offers many measures of such types of risk, these include: the expected shortfall measure of \cite{acharya2017measuring}; Co-Value-at-Risk (CoVaR) \citep{adrian2016covar}; and network connectedness measures \citep[see e.g.][]{demirer2018estimating}. The former measures relate to specific risk channels and as such can aid in calibrating regulatory tools. The latter approach permits one to quantify the overall influence of individual institutions to overall systemic risk, and hence identify systemically important financial institutions (SIFIs).  

Our approach relates closely with the above. However, under our framework one can characterise SIFIs, or important variables for different applications, throughout time as well as understanding whether the influence is persistent or transitory in nature. The benefit of this is twofold. First, enhancing prudential authorities' understanding of whether SIFIs influence are transient or long lasting can help refine policies in order to mitigate adverse firm behavior. For example, one could tailor policies by increasing the capital requirements of SIFIs for those who transmit persistent shocks that contribute significantly to systemic risk. Second, since systemic risk threatens the stability of the entire financial sector, knowing the frequency-specific sources of instability facilitates monitoring of changes to the such risks.

To illustrate how policy makers might use our approach, we examine the network structures of the financial sector at the granular level of 65 firms. Since our application uses daily data we have transitory and persistent network structures at every observation in our sample. We therefore focus on two different dates. Figure \ref{fig:app_network_day1} shows the network structures driven by transitory and persistent shocks for the SPF sector on October 24, 2008. This date corresponds to the start of the global financial crisis. Figure \ref{fig:app_network_day2} reports the corresponding network structures one year later on October 24, 2009. For each plot, arrows indicate the direction and strength of the connections, while a transparent (full colour) vertex indicates a stock that sends (receives) more shocks than it receives (sends). The size of the vertices indicates the net direction of the connections.

\begin{figure}[!ht]
\centering
\includegraphics[width=\textwidth]{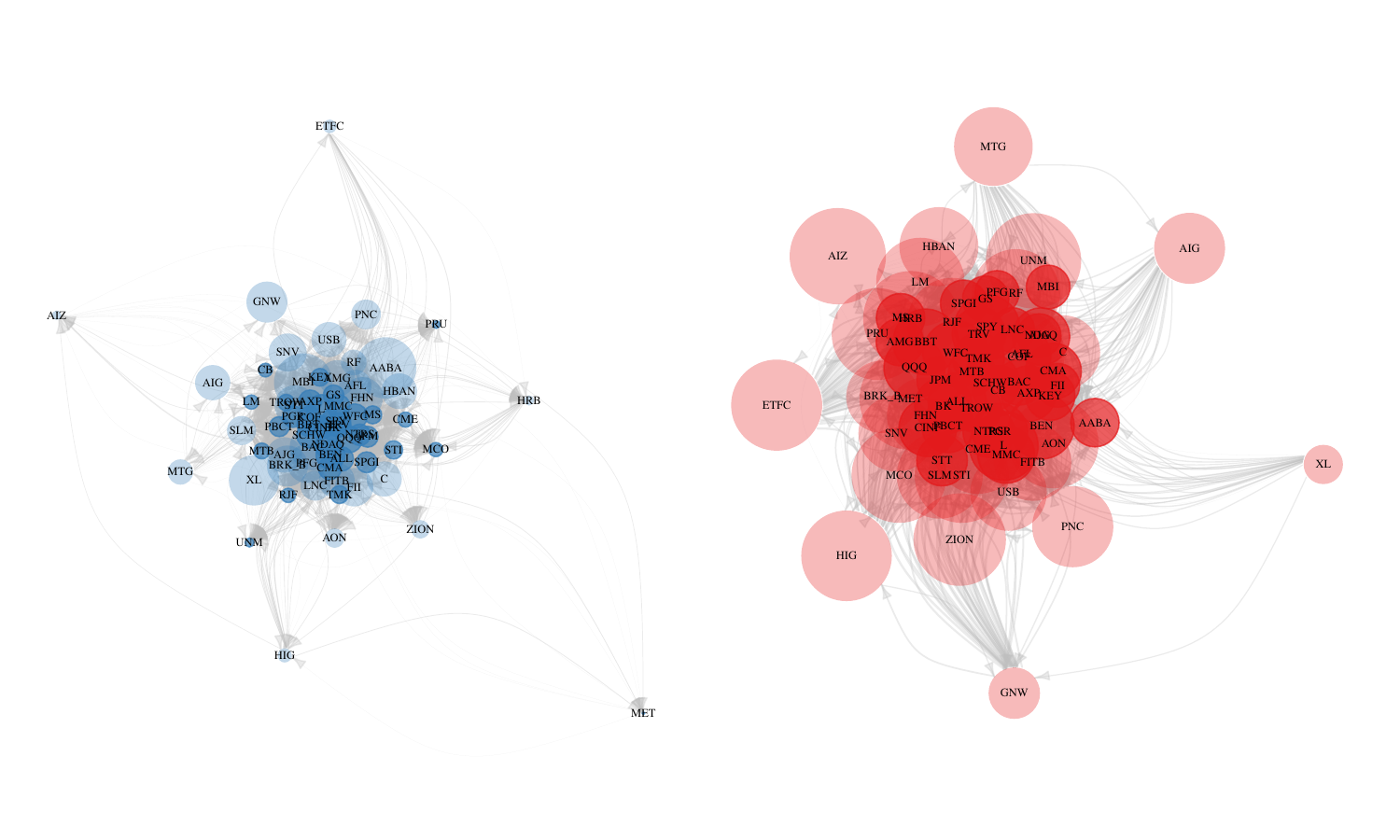}
    \caption{\textbf{Transitory and persistent networks of finance: 24 October 2008} \\ \small{The left (right) figure shows the network connections between the assets comprising the SPF sector driven by transitory (persistent) shocks on 24 October 2008, corresponding to the day when the VIX peaked. Arrows indicate the direction of the connections and the strength of the lines indicates the strength of the connections. Grey (black) vertices indicate firms that receive (send) more shocks than they send (receive). The size of the vertices indicates the net amount of shocks.}}
      \label{fig:app_network_day1}
\end{figure}

We can see that on 24 October 2008, the persistent links are larger relative to the transitory ones, suggesting that shocks within the financial sector create links that relate to the long-term. This suggests that systemic risk within the systems stems from persistent network structures. Now looking one year later, it is clear that connections are far weaker across both transitory and persistent network structures and systemic risk is relatively lower. The main takeaway from these plots is the strong differences in the overall structure of the horizon-specific networks. 

In Appendix D, we plot heatmaps showing the strength of financial institutions connections across persistent and transitory network structures for these same two dates in Figures D.2--D.3. Persistent shocks tend to drive the links with greater strength, and so we focus our discussion here. Zooming in to examine the contribution of specific firms, we can see that Truist Financial Corp (BBT), Franklin Resources (BEN), Loews Corporation (L), SPY, or Wells Fargo \& Co (WFC) transmitted persistent shocks to the financial sector and thus are identified as SIFIs that affected the system with persistent shocks. We can follow the contributions from the columns of the heatmap in Figure D.2. The impact of a bank increases with the number of rows containing a stronger and warmer red colour. As we can see, those banks we name above affect many other financial institutions at the start of the financial crisis, and such impacts are long-lasting. 

Conversely, Metlife (MET), Moody's (MCO), Unum (UNM), H\&R Block (HRB) and Assurant (AIZ) receiv the most shocks on 24 October 2008. One year later, we can see that the structure changes dramatically. While WFC, SPY and L seem to be strong SIFIs, BEN is a nearly non-contributing bank. This highlights how our approach tracks dynamics of key financial institutions within the system and their influence across persistent and transitory network structures. For completeness, we rank all institutions in the financial sector according to the strength of transitory and persistent shocks they transmit/receive during the same two dates in Table D.2 in Appendix D.

\begin{figure}[!ht]
\centering
\includegraphics[width=\textwidth]{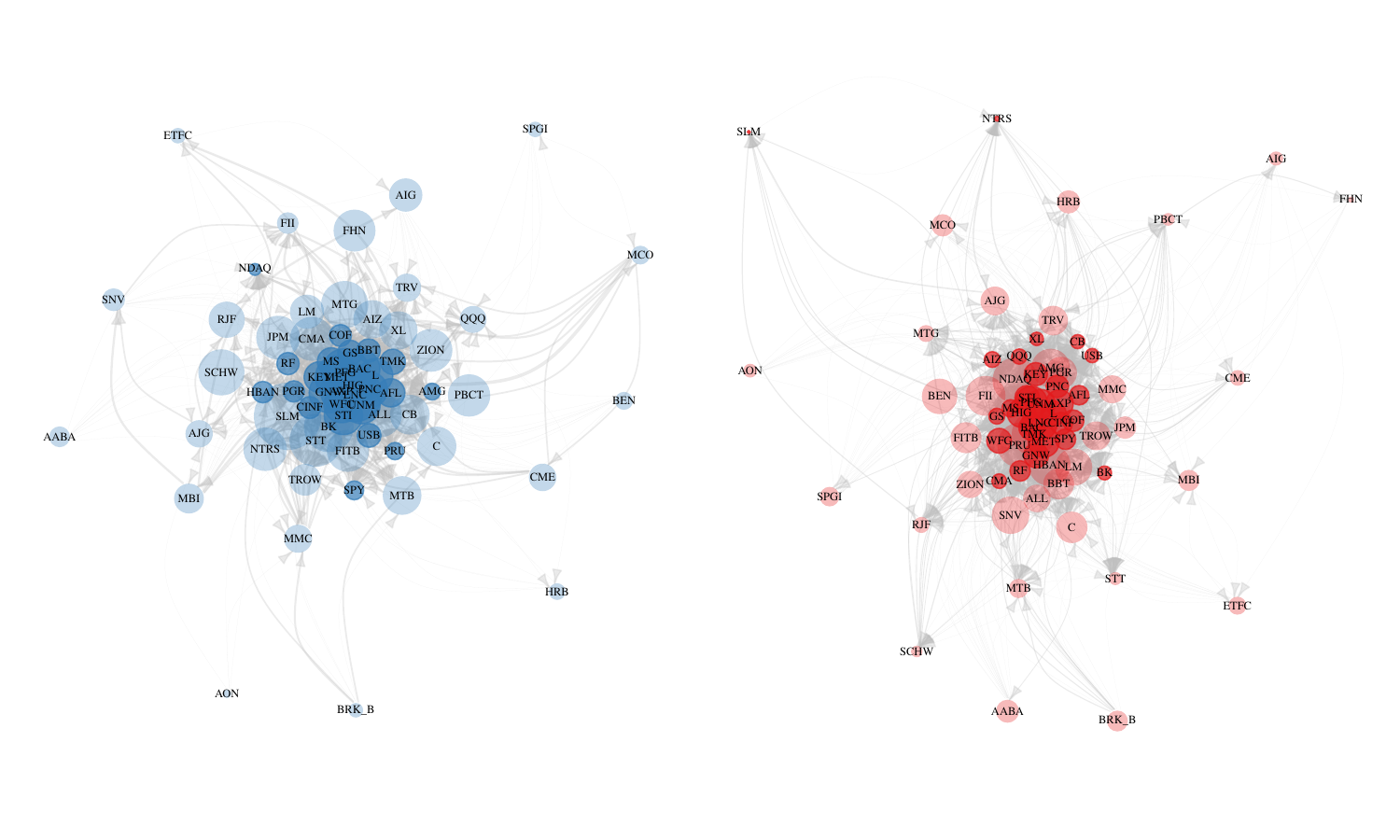}
    \caption{\textbf{Transitory and Persistent Network of Financials: 24 October 2009} \\ \small{The left (right) figure shows the network connections between the assets that make up the SPF sector by transitory (persistent) connectedness. Arrows indicate the direction of the connections and the strength of the lines indicates the strength of the connections. Grey (black) vertices denote firms that receive (send) more shocks than they send (receive). The size of the vertices indicates the net amount of shocks.}}
      \label{fig:app_network_day2}
\end{figure}

Overall, these network structures show how the role of a firm can change not only over time, but also in terms of persistence. Researchers may wish to use these network structures to assess the pricing implications of such risks. For example, users could assess the role of directional linkages in an empirical asset pricing application that builds on the theoretical work of \cite{branger2020equilibrium}. Our framework also allows one to track dynamic network structures that would complement studies such as \cite{herskovic2018networks} and \cite{gofman2020production}. The advantage of our approach is that one does not have to rely on monthly or annual data to capture such networks. Finally, looking at dynamic adjacency matrices can help economists understand how shocks dynamically determine network structures in models of monetary policy \citep[e.g.][]{pasten2020propagation} or the Phillips curve \citep[e.g.][]{rubbo2020networks}.

\section{Conclusion}\label{conclusion}
\label{sec:conclusion}

This paper proposes a novel framework for measuring dynamic network connections in a multivariate data system. We track dynamic connections driven by different degrees of persistence using a spectral decomposition of time-varying variance decomposition matrices. Our approach properly accounts for the characteristics of the shocks that create such links. We outline a procedure that allows one to test for statistical differences in connectedness over time and frequency. We provide Monte Carlo evidence that our measures are able to reliably track connectedness and correctly identify statistical differences from different data generating processes.

Empirically, we show that transitory and persistent measures of network connectedness improve our understanding of systemic risks arising from uncertainty networks. This is because our approach allows one to track connectedness across the transitory and persistent components of shocks. This is particularly useful during periods of heightened uncertainty, as our measures indicate whether systemic risks from network connections are transitory or persistent in nature. Ultimately, this could lead to better decision-making by macroprudential supervisors and investment decisions by market participants.
\\

\bibliographystyle{chicago}
\bibliography{bibliography}

\end{document}